\documentclass[letterpaper, 10pt, conference]{ieeeconf}
\let\proof\relax

\usepackage{algorithm}

\usepackage[T1]{fontenc}
\usepackage{amsthm,xpatch}
\usepackage{amsmath,amsfonts}
\usepackage{cite}
\usepackage{amssymb}
\usepackage{dsfont}
\usepackage{graphicx, subfigure}
\usepackage{color}
\usepackage{breqn}
\usepackage{mathtools}
\usepackage{bbm}
\usepackage{latexsym}
\usepackage{accents}
\usepackage{tikz}
\usepackage[mathscr]{eucal}
\usepackage{caption}
\usepackage{booktabs}
\usetikzlibrary{calc}
\usetikzlibrary{math}
\usetikzlibrary{arrows.meta}
\usepackage{csquotes}
\usepackage{varioref}
\newtheorem{lemma}{Lemma}
\newtheorem{theorem}{Theorem}

\newtheorem{example}{Example}

\newcommand{\mbf}[1]{\mathbf{#1}}
\newcommand*{\transpose}{%
  {\mathpalette\@transpose{}}%
}

\IEEEoverridecommandlockouts

\begin{document}

\newcommand{\SB}[3]{
\sum_{#2 \in #1}\biggl|\overline{X}_{#2}\biggr| #3
\biggl|\bigcap_{#2 \notin #1}\overline{X}_{#2}\biggr|
}

\newcommand{\Mod}[1]{\ (\textup{mod}\ #1)}

\newcommand{\overbar}[1]{\mkern 0mu\overline{\mkern-0mu#1\mkern-8.5mu}\mkern 6mu}

\makeatletter
\newcommand*\nss[3]{%
  \begingroup
  \setbox0\hbox{$\m@th\scriptstyle\cramped{#2}$}%
  \setbox2\hbox{$\m@th\scriptstyle#3$}%
  \dimen@=\fontdimen8\textfont3
  \multiply\dimen@ by 4             
  \advance \dimen@ by \ht0
  \advance \dimen@ by -\fontdimen17\textfont2
  \@tempdima=\fontdimen5\textfont2  
  \multiply\@tempdima by 4
  \divide  \@tempdima by 5          
  \ifdim\dimen@<\@tempdima
    \ht0=0pt                        
    \@tempdima=\fontdimen5\textfont2
    \divide\@tempdima by 4          
    \advance \dimen@ by -\@tempdima 
    \ifdim\dimen@>0pt
      \@tempdima=\dp2
      \advance\@tempdima by \dimen@
      \dp2=\@tempdima
    \fi
  \fi
  #1_{\box0}^{\box2}%
  \endgroup
  }
\makeatother

\makeatletter
\renewenvironment{proof}[1][\proofname]{\par
  \pushQED{\qed}%
  \normalfont \topsep6\p@\@plus6\p@\relax
  \trivlist
  \item[\hskip\labelsep
        \itshape
    #1\@addpunct{:}]\ignorespaces
}{%
  \popQED\endtrivlist\@endpefalse
}
\makeatother

\makeatletter
\newsavebox\myboxA
\newsavebox\myboxB
\newlength\mylenA

\newcommand*\xoverline[2][0.75]{%
    \sbox{\myboxA}{$\m@th#2$}%
    \setbox\myboxB\null
    \ht\myboxB=\ht\myboxA%
    \dp\myboxB=\dp\myboxA%
    \wd\myboxB=#1\wd\myboxA
    \sbox\myboxB{$\m@th\overline{\copy\myboxB}$}
    \setlength\mylenA{\the\wd\myboxA}
    \addtolength\mylenA{-\the\wd\myboxB}%
    \ifdim\wd\myboxB<\wd\myboxA%
       \rlap{\hskip 0.5\mylenA\usebox\myboxB}{\usebox\myboxA}%
    \else
        \hskip -0.5\mylenA\rlap{\usebox\myboxA}{\hskip 0.5\mylenA\usebox\myboxB}%
    \fi}
\makeatother

\xpatchcmd{\proof}{\hskip\labelsep}{\hskip3.75\labelsep}{}{}

\pagestyle{plain}

\title{\fontsize{22.59}{28}\selectfont Non-adaptive Quantitative Group Testing Using Irregular Sparse Graph Codes
}

\author{Esmaeil Karimi, Fatemeh Kazemi, Anoosheh Heidarzadeh, Krishna R. Narayanan, and Alex Sprintson\thanks{The authors are with the Department of Electrical and Computer Engineering, Texas A\&M University, College Station, TX 77843 USA (E-mail: \{esmaeil.karimi, fatemeh.kazemi, anoosheh, krn, spalex\}@tamu.edu).}\thanks{This material is based upon work supported by the National Science Foundation under Grants No. 1718658, 1642983, and 1547447.}}


\maketitle 

\thispagestyle{plain}

\begin{abstract}
This paper considers the problem of Quantitative Group Testing (QGT) where there are some defective items among a large population of $N$ items. We consider the scenario in which each item is defective with probability $K/N$, independently from the other items. In the QGT problem, the goal is to identify all or a sufficiently large fraction of the defective items by testing groups of items, with the minimum possible number of tests. 
In particular, the outcome of each test is a non-negative integer which indicates the number of defective items in the tested group. In this work, we propose a non-adaptive QGT scheme for the underlying randomized model for defective items, which utilizes sparse graph codes over {irregular} bipartite graphs with optimized degree profiles on the left nodes of the graph as well as binary $t$-error-correcting BCH codes. We show that in the sub-linear regime, i.e., when the ratio $K/N$ vanishes as $N$ grows unbounded, the proposed scheme with ${m=c(t,d)K(t\log (\frac{\ell N}{c(t,d)K}+1)+1)}$ tests can identify all the defective items with probability approaching $1$, where $d$ and $\ell$ are the maximum and average left degree, respectively, and $c(t,d)$ depends only on $t$ and $d$ (and does not depend on $K$ and $N$). For any $t\leq 4$, the testing and recovery algorithms of the proposed scheme have the computational complexity of $\mathcal{O}(N\log \frac{N}{K})$ and $\mathcal{O}(K\log \frac{N}{K})$, respectively. The proposed scheme outperforms two recently proposed non-adaptive QGT schemes for the sub-linear regime, including our scheme based on regular bipartite graphs and the scheme of Gebhard \emph{et al.}, in terms of the number of tests required to identify all defective items with high probability. 
\end{abstract}

\section{introduction}
We consider the Quantitative Group Testing (QGT) problem which is concerned with recovering all or a sufficiently large fraction of defective items in a given population of items, each of which is either defective or not. In the QGT problem, the result of a test on any group of items reveals the number of defective items in the tested group. The objective is to design a test plan for QGT with minimum number of tests. 

There are two different models for the defective items in the literature: \emph{deterministic} and \emph{randomized}. In the deterministic model (a.k.a. the combinatorial model), the exact number of defective items is known, whereas in the randomized model (a.k.a. the probabilistic model), each item is defective with some probability, independent of the other items \cite{sobel1966binomial,Mazumdar:2016:NGT:3024356.3025269,jiangnear,WZC:17,8437774}. In this work, we consider the randomized model in which each item is defective with probability $\frac{K}{N}$, independently from the other items, where $N$ is the total number of items, and the parameter $K$ represents the expected number of defective items. It should be noted that the deterministic model can be readily justified using the fact that performing one initial test on all items reveals the number of defective items. Notwithstanding, in most practical applications, performing a test on all items may not be feasible, particularly when the number of items is very large. On the other hand, assuming that the expected number of defective items is known is a more reasonable assumption for many practical applications. Moreover, it should be noted that the QGT schemes designed for the scenarios in which the randomized model is considered are applicable to the scenarios considering the deterministic model, but this relation does not work in reverse order.

In this paper, we are interested in \emph{non-adaptive} QGT schemes, where all tests are designed in advance. This is in contrast to \emph{adaptive} QGT schemes, in which the design of each test depends on the results of the previous tests. In most practical applications, when compared to adaptive QGT schemes, non-adaptive QGT schemes are preferred because all tests can be executed at once in parallel. 


\subsection{Related Work and Applications}
The QGT problem can be traced back to the seminal work by Shapiro in \cite{S:60}. To date, several adaptive and non-adaptive QGT strategies have been proposed, see, e.g.,~\cite{B:09,jiangnear,WZC:17,8437774,gebhard2019quantitative,DBLP:journals/corr/abs-1901-07635} and references therein. Using a simple information theoretic argument, one can easily show the information-theoretic lower bound $\log_K {N \choose K} \approx (K\log (N/K))/\log K$ on the minimum number of tests for any adaptive QGT scheme.\footnote{Throughout the paper the base of $\log$ is $2$, unless explicitly noted otherwise.} However, this lower bound is not tight for non-adaptive QGT schemes. In particular, it was shown in \cite{L:75} and \cite{djackov1975search} that any non-adaptive QGT scheme requires at least $(2K\log (N/K))/\log K$ tests. For the linear regime in which the number of defective items is a constant fraction of the total number of items, the QGT problem has been fully solved~\cite{NIPS2017_6641,8410892}. However, for the sub-linear regime, i.e., when the number of defective items grow sub-linearly in the total number of items, the QGT problem is widely open. 
Recently, in \cite{DBLP:journals/corr/abs-1901-07635}, we proposed the first non-adaptive QGT scheme for the sub-linear regime that requires ${m\approx 1.19K\log \left(4.74\frac{N}{K}\right)}$ tests to recover all the defective items with probability approaching $1$. Shortly after, Gebhard \emph{et al.} in \cite{gebhard2019quantitative} proposed a greedy non-adaptive QGT scheme that requires ${m=\frac{1+\sqrt{\theta}}{1-\sqrt{\theta}} K\ln \left(\frac{N}{K}\right)}$ tests to recover all $K = N^{\theta}$ (for $0<\theta<1$) defective items with high probability. 

Aside from the theoretical endeavors, the QGT problem has also gained substantial attention over the last few years from the practical perspective. In particular, the QGT problem has been studied for a wide range of applications from machine learning and computational biology \cite{acharya2019optimal,cao2014quantitative} to multi-access communication, traffic monitoring, and network tomography \cite{10.1007/978-3-030-25027-0_10,wang2015group,6121984}. It should be noted that most of these applications are being run repeatedly over time, and for such applications, minimizing the constant factor hidden in the order is also of prominent importance. This observation is the primary motivation for this work. 

\subsection{Main Contributions}
In this work, we propose a non-adaptive QGT scheme for the scenarios in which the randomized model is considered for defective items. The testing algorithm of the proposed scheme relies on sparse graph codes over {irregular} bipartite graphs with optimized left-degree profiles as well as binary $t$-error-correcting BCH codes. 
As part of the process of optimizing the left-degree profile of the graph, we take advantage of the  density-evolution technique to analyze the probability of error of the proposed peeling-based recovery algorithm, i.e., the probability that a defective item remains unidentified over the iterations of the recovery algorithm. We provide provable guarantees on the performance of the proposed scheme in terms of the required number of tests. In particular, we show that in the sub-linear regime the proposed scheme requires ${m=c(t,d)K(t\log (\frac{\ell N}{c(t,d)K}+1)+1)}$ tests to identify all defective items with high probability, where $d$ and $\ell$ are the maximum and average left degree, respectively, and $c(t,d)$ is constant with respect to $K$ and $N$, and depends only on $t$ and $d$. Moreover, we show that, for any $t\leq 4$, the testing and recovery algorithms of the proposed scheme have the computational complexity of $\mathcal{O}(N\log \frac{N}{K})$ and $\mathcal{O}(K\log \frac{N}{K})$, respectively.

\section{Problem Setup and Notations}\label{sec:SN}
Throughout the paper, we denote vectors and matrices by bold-face small and capital letters, respectively. For an integer $i\geq 1$, we denote $\{1,\dots,i\}$ by $[i]$.

In this work, we consider a quantitative group testing (QGT) problem with a randomized model for defective items, where in a population of $N$ items, each item is defective with probability $\frac{K}{N}$, independently from the other items. The problem is to identify all or a sufficiently large fraction of the defective items by testing groups of items, with the minimum possible number of tests, where the outcome of each test is a non-negative integer that indicates the number of defective items in the tested group. The focus of this work is on the sub-linear regime where the parameter $K$ grows sub-linearly with the total number of items ($N$).

 We define the support vector $\mathbf{x}\in \{0,1\}^N$ to represent the set of $N$ items. The $i$-th component of $\mathbf{x}$ is $1$ if and only if the $i$-th item is defective. In a non-adaptive QGT problem, designing a test scheme consisting of $m$ tests is equivalent to the construction of a binary matrix with $m$ rows which is referred to as measurement matrix. We let matrix ${\textbf{A}\in \{0,1\}^{m\times N}}$ denote the measurement matrix wherein the non-zero indices in the $i$-th row correspond to the items that are present in the $i$-th test. We also let vector $\mathbf{y}\in\{ 0,1,2,\dots\}^m$ denote the outcomes of the $m$ tests in the following matrix form.
\begin{equation}\label{eq:gtresult}
\mathbf{y}=[y_{1},\dots,y_{m}]^{\mathsf{T}}=\mathbf{A}\mathbf{x}.
\end{equation}

The objective is to construct a measurement matrix with a small number of rows (tests) that successfully identifies the set of defective items with high probability given the test results vector $\mathbf{y}$.

\section{Proposed Algorithm \label{sec:main results}}

\subsection{Testing algorithm}

We employ a framework similar to that proposed in \cite{DBLP:journals/corr/abs-1901-07635} for designing the measurement matrix $\textbf{A}$; however, in  our design we utilize {irregular} bipartite graphs with carefully designed left-degree profile, instead of bi-regular bipartite graphs.

Consider a randomly generated bipartite graph with $N$ left nodes and $M$ right nodes where each right node is connected to $r$ left nodes. The left nodes are connected to the right nodes according to a left-node degree distribution given by ${L(x)\triangleq \sum_{i=1}^{d}L_ix^{i}}$ where $d$ and $L_i$ denote the maximum degree of a left node and the probability that a randomly selected left node in the graph has degree $i$, respectively. We denote the adjacency matrix of such a graph by ${\mathbf{T} \in \{0,1\}^{M\times N}}$ where each column in $\mathbf{T}$ corresponds to a left node, and each row in $\mathbf{T}$ corresponds to a right node and has exactly $r$ ones. The adjacency matrix $\mathbf{T}$ can be represented in the matrix form $\mathbf{T}=[\mathbf{t}_{1}^{\mathsf{T}},\mathbf{t}_{2}^{\mathsf{T}},\dots,\mathbf{t}_{M}^{\mathsf{T}}]^{\mathsf{T}}$, where $\mathbf{t}_{i}$ denotes the $i$-th row.

A carefully designed signature matrix $\mathbf{U}\in \{0,1\}^{s\times r}$ is used to assign $s$ tests to each right node. We place an all-ones row of length $r$ as the first row of the signature matrix. The first row in $\mathbf{U}$ corresponds to a test whose result reveals the number of defective items connected to a right node. The rest of the rows in $\mathbf{U}$ are the rows in the parity-check matrix of a binary $t$-error-correcting BCH code \cite{lin2001error}. Given that the number of defective items connected to a right node is no more than $t$, the results of the tests corresponding to the rows in the parity-check matrix can be used to identify the defective items connected to the right node. Considering that the number of columns is $r$, the number of rows in the parity-check matrix of a $t$-error-correcting BCH code is given by $R=t\log ({r+1})$. The signature matrix $\mathbf{U}$ can then be represented by ${\mathbf{U}=[\mathbf{1}_{1\times r}^{\mathsf{T}} ,\mathbf{H}_t^{\mathsf{T}}]^{\mathsf{T}}}$, where $\mathbf{1}_{1\times r} $ is an all-ones row of length $r$, and ${\mathbf{H}_t \in \{0,1\}^{R\times r}}$ is the parity-check matrix of a binary $t$-error-correcting BCH code. One can readily observe that the number of rows in $\mathbf{U}$ is given by ${s=R+1=t\log (r+1)+1}$. 

Now, we show the construction process of the measurement matrix using the adjacency matrix $\mathbf{T}$ and the signature matrix $\mathbf{U}$. Let the measurement matrix be given by ${\mathbf{A}=[\mathbf{A}_{1}^{\mathsf{T}},\dots,\mathbf{A}_{M}^{\mathsf{T}}]^{\mathsf{T}}}$ where ${\mathbf{A}_i\in \{0,1\}^{s\times N}}$ is a block matrix that represents the $s$ tests at the $i$-th right node. Let $\mathbf{u}_j$ denote the $j$-th column of the signature matrix. Note that the number of columns in the signature matrix $\mathbf{U}$ is $r$, and there are exactly $r$ ones in each row of the adjacency matrix $\mathbf{T}$. The block matrix $\mathbf{A}_i$ is then constructed by replacing zeros and ones in the $i$-th row of the adjacency matrix, $\mathbf{t}_i$, by all-zero columns and the columns of the signature matrix, respectively, as follows:
\begin{align}\label{eq:measureblock}
 \mbf{A}_i=[\mbf{0},\ldots,\mbf{0},\mbf{u}_1, \mbf{0},\ldots, \mbf{u}_2,\mbf{0}, \ldots, \mbf{u}_{r}]
 \end{align} where $\mbf{t}_i =[0,\ldots,0,\hspace{0.6ex}1,\hspace{0.9ex} 0, \ldots,\hspace{0.6ex}1,\hspace{0.9ex}0, \ldots, \hspace{0.9ex}1]$. 
In other words, we place the $r$ columns of the signature matrix at the coordinates of the $r$ ones in the row $\mathbf{t}_i$, and then we replace zeros in $\mathbf{t}_i$ by all-zero columns. The total number of rows in the measurement matrix $\mathbf{A}$ which is equivalent to the total number of tests in the proposed scheme is given by ${m=M\times s =M(t\log (r+1)+1)}$. The following example helps to better understand the construction process of the measurement matrix.

\begin{example}\label{ex:example1}
Let $\mathbf{T}$ denote the adjacency matrix of {an irregular} bipartite graph with $N=14$ left nodes and $M=3$ right nodes of degree $r=7$. The edge connections of the left side satisfies the following left node degree distribution given by ${L(x)=\frac{10}{14}x+ \frac{1}{14}x^{2}+\frac{3}{14}x^{3}}$.

\setcounter{MaxMatrixCols}{14}
\[
\mathbf{T}= \begin{bmatrix}
     0 &\textcolor{blue}{1} & 0 & \textcolor{lime}{1} & \textcolor{orange}{1} & 0 & 0 & 0 & \textcolor{green}{1} & \textcolor{red}{1} & 0 & 0 & \textcolor{brown}{1} & \textcolor{yellow}{1}\\
    0 & 0 & \textcolor{blue}{1} & \textcolor{lime}{1}  & 0 & 0 & \textcolor{orange}{1} & \textcolor{green}{1} & 0 & \textcolor{red}{1} & 0 & \textcolor{brown}{1} & \textcolor{yellow}{1} & 0\\
      \textcolor{blue}{1} & 0 & 0 & \textcolor{lime}{1} & 0 & \textcolor{orange}{1} & 0 & \textcolor{green}{1} & 0 & \textcolor{red}{1} & \textcolor{brown}{1} & 0 & \textcolor{yellow}{1} & 0\\
      
\end{bmatrix}.
\] 
Also, we let $\mbf{H}_1$ and ${\mathbf{U}=[\mathbf{1}_{1\times 7}^{\mathsf{T}},\mathbf{H}_1^{\mathsf{T}}]^{\mathsf{T}}}$ denote the parity-check matrix of a binary $t=1$-error-correcting BCH code of length $r=7$ and the signature matrix, respectively, 
 \[
         \mbf{H}_1 =[\mathbf{h}_1,\dots,\mathbf{h}_7]=  \begin{bmatrix}
0 & 0  & 1 & 0 & 1  & 1 & 1\\
 0  & 1 & 0 & 1  & 1 & 1 & 0\\
1 & 0 & 0 & 1 & 0 & 1 & 1
\end{bmatrix},
\]
 
\[
\mbf{U} = \begin{bmatrix}
 \textcolor{blue}{1} & \textcolor{lime}{1}  & \textcolor{orange}{1} &  \textcolor{green}{1} & \textcolor{red}{1}  & \textcolor{brown}{1} & \textcolor{yellow}{1}\\
\textcolor{blue}{0} & \textcolor{lime}{0}  & \textcolor{orange}{1} & \textcolor{green}{0} & \textcolor{red}{1}  & \textcolor{brown}{1} & \textcolor{yellow}{1}\\
\textcolor{blue}{0}  & \textcolor{lime}{1} & \textcolor{orange}{0} & \textcolor{green}{1}  & \textcolor{red}{1} & \textcolor{brown}{1} & \textcolor{yellow}{0}\\
\textcolor{blue}{1} & \textcolor{lime}{0} & \textcolor{orange}{0} & \textcolor{green}{1} & \textcolor{red}{0} & \textcolor{brown}{1} & \textcolor{yellow}{1}
\end{bmatrix}.
\] 
The measurement matrix $\mathbf{A}$ can then be constructed by following the procedure explained earlier,
\[
   \mathbf{A}= \begin{bmatrix}
 0 & \textcolor{blue}{1} & 0 & \textcolor{lime}{1} & \textcolor{orange}{1} & 0 & 0 & 0 & \textcolor{green}{1}  & \textcolor{red}{1} & 0 & 0 & \textcolor{brown}{1} & \textcolor{yellow}{1} \\
0 & \textcolor{blue}{0} & 0 & \textcolor{lime}{0} & \textcolor{orange}{1} & 0 & 0 & 0& \textcolor{green}{0} & \textcolor{red}{1} & 0 & 0 & \textcolor{brown}{1} & \textcolor{yellow}{1} \\
 0 & \textcolor{blue}{0} & 0 & \textcolor{lime}{1} & \textcolor{orange}{0} & 0 & 0 & 0& \textcolor{green}{1}  & \textcolor{red}{1} & 0 & 0& \textcolor{brown}{1} &  \textcolor{yellow}{0} \\
0 & \textcolor{blue}{1} & 0 & \textcolor{lime}{0} & \textcolor{orange}{0} & 0 & 0 & 0& \textcolor{green}{1} & \textcolor{red}{0} & 0 & 0& \textcolor{brown}{1} & \textcolor{yellow}{1} \\ \hline
           
0 & 0 &\textcolor{blue}{1} & \textcolor{lime}{1} &0 & 0 & \textcolor{orange}{1} & \textcolor{green}{1} & 0 & \textcolor{red}{1} & 0 & \textcolor{brown}{1} & \textcolor{yellow}{1} & 0\\
0 & 0 &\textcolor{blue}{0} & \textcolor{lime}{0} & 0 & 0 &\textcolor{orange}{1} & \textcolor{green}{0} & 0 & \textcolor{red}{1} & 0 & \textcolor{brown}{1} & \textcolor{yellow}{1} & 0 \\
0 & 0 &\textcolor{blue}{0} & \textcolor{lime}{1} & 0 & 0 &\textcolor{orange}{0} & \textcolor{green}{1} & 0 & \textcolor{red}{1} & 0 & \textcolor{brown}{1} & \textcolor{yellow}{0} & 0\\
0 & 0 &\textcolor{blue}{1} & \textcolor{lime}{0} & 0 & 0 &\textcolor{orange}{0} & \textcolor{green}{1} & 0 & \textcolor{red}{0} & 0 & \textcolor{brown}{1} & \textcolor{yellow}{1} & 0\\ \hline 
           
 \textcolor{blue}{1} & 0 & 0 & \textcolor{lime}{1} &0 & \textcolor{orange}{1} & 0 &\textcolor{green}{1} & 0 & \textcolor{red}{1} & \textcolor{brown}{1} & 0 & \textcolor{yellow}{1} & 0\\
 \textcolor{blue}{0} & 0 & 0 & \textcolor{lime}{0} &0 & \textcolor{orange}{1} & 0 &\textcolor{green}{0} & 0 & \textcolor{red}{1} & \textcolor{brown}{1} & 0 & \textcolor{yellow}{1} & 0\\
 \textcolor{blue}{0} & 0 & 0 & \textcolor{lime}{1} &0 & \textcolor{orange}{0} & 0 &\textcolor{green}{1} & 0 & \textcolor{red}{1} & \textcolor{brown}{1} & 0 & \textcolor{yellow}{0} & 0\\
 \textcolor{blue}{1} & 0 & 0 & \textcolor{lime}{0} &0 & \textcolor{orange}{0} & 0 & \textcolor{green}{1} & 0 & \textcolor{red}{0} & \textcolor{brown}{1} & 0 & \textcolor{yellow}{1} & 0
       
         \end{bmatrix}.
\]
\end{example}
 
\subsection{Recovery Algorithm}\label{decoding}

The recovery algorithm is similar to the peeling decoding algorithm, and it proceeds in an iterative manner as follows. During each iteration, the recovery algorithm inspects all the right nodes, and identifies and resolves any right node which is connected to $t$ or less number of defective items (for more details, see the proof of \cite[Lemma 1]{DBLP:journals/corr/abs-1901-07635}). Then, the recovery algorithm peels the edges connected to the identified defective items off the graph, and the next iteration begins. When no (not-yet-resolved) right node connected to $t$ or less number of defective items can be found, the recovery algorithm terminates. Below, we provide an illustrative example of the recovery algorithm.
\begin{example}
Consider the scenario in Example~\ref{ex:example1}. Suppose that items $4$,$8$, and $11$ are defective. Let the support vector $\mathbf{x}=[0,0,0,1,0,0,0,1,0,0,1,0,0,0]^T$ represent the set of $N=14$ items. The test results vector $\mathbf{y}$ according to the testing algorithm using the measurement matrix $\mathbf{A}$ constructed in Example~\ref{ex:example1} can be expressed as follows:
\begin{align*}
   \mathbf{y}=[y_1,\cdots,y_{12}]^{\mathsf{T}}=\mathbf{A}\mathbf{x}=
    \begin{bmatrix}
   \mathbf{u}_2  \\
   \mathbf{u}_{2}+ \mathbf{u}_{4} \\
   \mathbf{u}_{2}+\mathbf{u}_{4} +\mathbf{u}_{6}\\
    \end{bmatrix}.
\end{align*}
The results of the tests corresponding to the right nodes $1,2,3$ are respectively given by
\[ [y_1,y_2,y_3,y_4]^{\mathsf{T}}=\mathbf{u}_2=[1,0,1,0]^{\mathsf{T}},\]
\[ [y_5,y_6,y_7,y_8]^{\mathsf{T}}= \mathbf{u}_{2}+ \mathbf{u}_{4}=[2,0,2,1]^{\mathsf{T}},\]
\[ [y_9,y_{10},y_{11},y_{12}]^{\mathsf{T}}=\mathbf{u}_2+\mathbf{u}_{4} +\mathbf{u}_{6}=[3,1,3,2]^{\mathsf{T}}.\]

Since we used the parity-check matrix of a $t=1$-error-correcting BCH code to build the signature matrix, each right node can be resolved (i.e., all items connected to the right node can be identified) if it is connected to at most one defective item. The first test result associated to a right node shows the number of defective items connected to that right node. In the first iteration, the decoding algorithm can only resolve the first right node because $y_1=1$ and $y_5,y_9\neq 1$. Using $[y_2,y_3,y_4]^{\mathsf{T}}=\mathbf{h}_2=[0,1,0]^{\mathsf{T}}$, by using a BCH decoding algorithm we can identify the second item connected to the first right node, i.e., item $4$, as a defective item. Subtracting off the contribution of the item $4$ from the test results corresponding to the unresolved right nodes, the updated test results will be as follows:
\[ [y_5,y_6,y_7,y_8]^{\mathsf{T}}= \mathbf{u}_{4}=[1,0,1,1]^{\mathsf{T}}\]
\[ [y_9,y_{10},y_{11},y_{12}]^{\mathsf{T}}=\mathbf{u}_{4} +\mathbf{u}_{6}=[2,1,2,2]^{\mathsf{T}}\]

In the second iteration, the recovery algorithm resolves the second right node because $y_5=1$ and $y_9\neq 1$. A BCH decoding algorithm uses $[y_6,y_7,y_8]^{\mathsf{T}}=\mathbf{h}_4=[0,1,1]^{\mathsf{T}}$, and declares the forth item connected to the second right node, i.e., item $8$, as a defective item. Similarly as in the case of item $4$ in the first iteration, subtracting off the contribution of the item $8$ from the test results corresponding to the unresolved right nodes, the updated test results will be as follows:
\[ [y_9,y_{10},y_{11},y_{12}]^{\mathsf{T}}=\mathbf{u}_{6}=[1,1,1,1]^{\mathsf{T}}\]
Since $y_9=1$, the recovery algorithm is then able to resolve the third right node in the third iteration. Looking at ${[y_{10},y_{11},y_{12}]^{\mathsf{T}}=\mathbf{h}_{6}=[1,1,1]^{\mathsf{T}}}$, by using a BCH decoding algorithm we can identify the sixth item connected to the third right node, i.e., item $11$, as a defective item. Since all $3$ right nodes are resolved, the recovery algorithm cannot find any not-yet-resolved right node (connected to $1$ or less defective items), and hence the recovery algorithm terminates. For this example, the recovery algorithm successfully identified all $3$ defective items. 
\end{example}

\section{Main Results}
We present our main results in this section. Theorem~\ref{theo:main1} specifies the number of tests required by the proposed QGT scheme in the sub-linear regime. Theorem~\ref{thm:main2} states the computational complexity of the testing and recovery algorithms of the proposed QGT scheme. The proofs of Theorem~\ref{theo:main1} and~\ref{thm:main2} are given in Section~\ref{sec:proofs}.

\begin{theorem}\label{theo:main1}
In the sub-linear regime, the proposed QGT scheme requires ${m=c(t,d)K(t\log (\frac{\ell N}{c(t,d)K}+1)+1)}$ tests to identify all defective items with probability approaching $1$, where $d$ and $\ell$ are the maximum and average left degree, respectively; and $c(t,d)$ is constant in $N$ and $K$, and depends only on $t$ and $d$. Table~\ref{Table:constant1} shows the values of $c(t,d)$ for $t=1$ and ${d\in \{3,4,\cdots,18\}}$, and Table~\ref{Table:constant2} (or respectively, Table~\ref{Table:constant3}) shows the values of $c(t,d)$ for $t=2$ (or respectively, $t=3$) and ${d\in \{2,3,\cdots,17\}}$. 

\begin{theorem}\label{thm:main2}
For any $t\leq 4$, the testing and recovery algorithms of the proposed QGT scheme have the computational complexity of $\mathcal{O}(N\log \frac{N}{K})$ and $\mathcal{O}(K\log \frac{N}{K})$, respectively.
\end{theorem}

\begin{table*}[t]
\centering
\caption{\small{The constant $c(t,d)$ for $t=1$ and ${d\in \{3,4,\cdots,18\}}$.}}
\begin{tabular}{| c | c | c | c |c | c | c | c | c | c | c | c | c | c | c | c |c |}
\hline
$d$ & 3 & 4 & 5 & 6  & 7 & 8 & 9 & 10 & 11 & 12 & 13 & 14 & 15 & 16 & 17 & 18\\ \hline
$\lambda_2$ & & & & & & & & & & & & & & & & \\ \hline
$\lambda_3$ & 1 & 0.785 & 0.765 & 0.746 & 0.723 & 0.705 & 0.69 & 0.676 & 0.658 & 0.646 & 0.634 & 0.621 & 0.611 & 0.595 & 0.579 & 0.564\\ \hline
$\lambda_4$ &  & 0.215 & &  & &  & & &  & & & &   & & & \\\hline
$\lambda_5$ &  & & 0.235 &  & &  & & &  & & & &  & & & \\ \hline
$\lambda_6$ &  &  & & 0.254 & &  & & &  & & & &   & & & \\ \hline
$\lambda_7$ &  &  & &  & 0.277 &  & & &  & & & &   & & & \\ \hline
$\lambda_8$ &  &  & &  & & 0.295 & & &  & & & &  & & & \\ \hline
$\lambda_9$ &  &  & &  & &  & 0.31 & &  & & & &   & & & \\ \hline
$\lambda_{10}$ &  & & &  & &  & & 0.324 &  & & & &   & & & \\ \hline
$\lambda_{11}$ &  & & &  & &  & & & 0.342 & & & &  & & & \\ \hline
$\lambda_{12}$ &  & & &  & &  & & &  & 0.354 & & &   & & & \\ \hline
$\lambda_{13}$ &  & & &  & &  & & &  & & 0.366 & &  & & & \\ \hline
$\lambda_{14}$ &  &  & &  & &  & & &  & & & 0.379 &  & & & \\ \hline
$\lambda_{15}$ &  & & &  & &  & & &  & & & &  0.389 & & & \\ \hline
$\lambda_{16}$ &  & & &  & &  & & &  & & & &  & 0.405 & 0.005 & \\ \hline
$\lambda_{17}$ &  & & &  & &  & & &  & & & &  & & 0.416 & 0.003\\ \hline
$\lambda_{18}$ &  & & &  & &  & & &  & & & &  & &  & 0.433\\ \hline
 $\ell$ & 3 & 3.17 & 3.312 & 3.437 & 3.563 & 3.678 & 3.783 & 3.88 & 3.993 & 4.084 & 4.177 & 4.273 & 4.356 & 4.473 & 4.592 & 4.709\\ \hline
 $c(t,d)$ & 1.222 & 1.217 & 1.208 & 1.197 & 1.186 & 1.175 & 1.164 & 1.153 & 1.142 & 1.133 & 1.123 & 1.114 & 1.106 & 1.098 & 1.093 & 1.09\\ \hline
\end{tabular} 
\label{Table:constant1}
\end{table*}

\begin{table*}[t]
\centering
\vspace{0.5cm}
\caption{\small{The constant $c(t,d)$ for $t=2$ and ${d\in \{2,3,\cdots,17\}}$.}}
\begin{tabular}{| c | c | c | c |c | c | c | c | c | c | c | c | c | c | c | c | c |}
\hline
$d$ & 2 & 3 & 4 & 5 & 6  & 7 & 8 & 9 & 10 & 11 & 12 & 13 & 14 & 15 & 16 & 17\\ \hline
$\lambda_2$ & 1 & 0.659 & 0.69 & 0.681 & 0.666 & 0.653 & 0.639 & 0.619 & 0.592 & 0.57 & 0.56 & 0.554 & 0.549 & 0.546 & 0.541 & 0.536\\ \hline
$\lambda_3$ &  & 0.341 &  & & &  & &  & & &  & & & & & \\ \hline
$\lambda_4$ &  &  & 0.31 &  & &  & & &  & & & &  & & &\\\hline
$\lambda_5$ &  & &  & 0.319 & &  & & &  & & & &  & & &\\ \hline
$\lambda_6$ &  &  & &  & 0.334 &  & & &  & & & &  & & &\\ \hline
$\lambda_7$ &  &  & &  &  & 0.347 & & &  & 0.001 & 0.049 & 0.09 & 0.059 & 0.022 & 0.001 &\\ \hline
$\lambda_8$ &  &  & &  & & & 0.361 & &  & & & 0.004 &  0.074 & 0.144 & 0.187 & 0.199\\ \hline
$\lambda_9$ &  &  & &  & &  &  & 0.381 & 0.002 & & & &  & & & \\ \hline
$\lambda_{10}$ &  & & &  & &  & &  & 0.406 & & & &  & & & \\ \hline
$\lambda_{11}$ &  & & &  & &  & & &  & 0.429 & & &  & & & \\ \hline
$\lambda_{12}$ &  & & &  & &  & & &  & & 0.391 & &  & & & \\ \hline
$\lambda_{13}$ &  & & &  & &  & & &  & &  & 0.352 &  & & & \\ \hline
$\lambda_{14}$ &  &  & &  & &  & & &  & & &  & 0.317 & & & \\ \hline
$\lambda_{15}$ &  &  & &  & &  & & &  & & &  &  & 0.288 & & \\ \hline
$\lambda_{16}$ &  &  & &  & &  & & &  & & &  &  & & 0.271 & \\ \hline
$\lambda_{17}$ &  &  & &  & &  & & &  & & &  &  & & & 0.265\\ \hline
 $\ell$ & 2 & 2.257 & 2.367 & 2.474 & 2.573 & 2.659 & 2.741 & 2.843 & 2.969 & 3.085 & 3.126 & 3.15 & 3.174 & 3.193 & 3.214 & 3.242\\ \hline
 $c(t,d)$ & 0.597 & 0.582 & 0.572 & 0.562 & 0.553 & 0.545 & 0.538 & 0.531 & 0.528 & 0.527 & 0.526 & 0.526 & 0.526 & 0.525 & 0.525 & 0.525\\ \hline
\end{tabular} 
\label{Table:constant2}
\end{table*}

\begin{table*}[t]
\centering
\vspace{0.5cm}
\caption{\small{The constant $c(t,d)$ for $t=3$ and ${d\in \{2,3,\cdots,17\}}$.}}
\begin{tabular}{| c | c | c | c |c | c | c | c | c | c | c | c | c | c | c | c | c |}
\hline
$d$ & 2 & 3 & 4 & 5 & 6  & 7 & 8 & 9 & 10 & 11 & 12 & 13 & 14 & 15 & 16 & 17\\ \hline
$\lambda_2$ & 1 & 0.97 & 0.889 & 0.844 & 0.807 & 0.784 & 0.759 & 0.737 & 0.72 & 0.704 & 0.686 & 0.668 & 0.653 & 0.639 & 0.632 & 0.63\\ \hline
$\lambda_3$ &  & 0.03 &  & & &  & &  & & &  & &  &  & &\\ \hline
$\lambda_4$ &  &  & 0.111 &  & &  & & &  & & & &  &  & &\\\hline
$\lambda_5$ &  & &  & 0.156 & &  & & &  & & & &  &  & &\\ \hline
$\lambda_6$ &  &  & &  & 0.193 &  & & &  & & & &  &  & &\\ \hline
$\lambda_7$ &  &  & &  &  & 0.216 & & &  & &  &  &  &  & &\\ \hline
$\lambda_8$ &  &  & &  & & & 0.241 & &  & & &  &  &  & &\\ \hline
$\lambda_9$ &  &  & &  & &  &  & 0.263 &  & & & &  &  & &\\ \hline
$\lambda_{10}$ &  & & &  & &  & &  & 0.28 & & & &  &  & &\\ \hline
$\lambda_{11}$ &  & & &  & &  & & &  & 0.296 & & &  &  & &\\ \hline
$\lambda_{12}$ &  & & &  & &  & & &  & & 0.314 & &  &  & &\\ \hline
$\lambda_{13}$ &  & & &  & &  & & &  & &  & 0.332 & 0.001 &  & 0.045 & 0.11\\ \hline
$\lambda_{14}$ &  &  & &  & &  & & &  & & &  & 0.346 &  & & \\ \hline
$\lambda_{15}$ &  &  & &  & &  & & &  & & &  &  & 0.361 & & \\ \hline
$\lambda_{16}$ &  &  & &  & &  & & &  & & &  &  &  & 0.323 & \\ \hline
$\lambda_{17}$ &  &  & &  & &  & & &  & & &  & &  & & 0.26\\ \hline
 $\ell$ & 2 & 2.021 & 2.118 & 2.207 & 2.295 & 2.366 & 2.442 & 2.515 & 2.577 & 2.639 & 2.709 & 2.781 & 2.848 & 2.909 & 2.945 & 2.952\\ \hline
 $c(t,d)$ & 0.388 & 0.388 & 0.387 & 0.384 & 0.381 & 0.378 & 0.375 & 0.372 & 0.37 & 0.367 & 0.365 & 0.363 & 0.363 & 0.362 & 0.362 & 0.362\\ \hline
\end{tabular} 
\label{Table:constant3}
\end{table*}
\end{theorem}

\section{Proof of Main Theorems}\label{sec:proofs}
\subsection{Proof of Theorem~\ref{theo:main1}}
Consider a group of $N$ items where each item is defective with probability $\gamma\triangleq \frac{K}{N}$. Also, consider {an irregular} bipartite graph with $N$ left nodes and $M$ right nodes where each right node is connected to $r$ left nodes. The left nodes are connected to the right nodes according to a left-node degree distribution given by ${L(x)= \sum_{i=1}^{d}L_ix^{i}}$ where $d$ and $L_i$ denote the maximum degree of a left node and the probability that a randomly selected left node in the graph has degree $i$, respectively. The average left degree can be computed by ${\ell=\sum_{i=1}^{d}iL_i}$. Since the number of edges connected to the left nodes is equal to the number of edges connected to the right nodes, the following equation holds.
\begin{equation}\label{eq:edgeBalance}
    N\ell=Mr
\end{equation}
The left edge degree distribution can be defined by ${\lambda(x)\triangleq \sum_{i=1}^{d}\lambda_ix^{i-1}}=\frac{L'(x)}{L'(1)}$ where $\lambda_i$ denotes the probability that a randomly selected edge in the graph is connected to a left node of degree $i$. It is easy to see that $L'(1)=\ell$. Thus, one can readily compute $\lambda_i=\frac{iL_i}{\ell}$. Using the fact that $\sum_{i=1}^{d}L_i=1$, we can rewrite the last equation as follows.
\begin{equation}\label{eq:avg-lft-dgr}
    \frac{1}{\ell}=\sum_{i=1}^{d}\frac{\lambda_i}{i}
\end{equation}
 
 We leverage the density evolution technique to analyze the fraction of defective items remains unidentified at the end of each iteration of the recovery algorithm.
 
 \begin{figure}
\vspace{0.5cm}
\begin{tikzpicture}
 
  \coordinate (O1) at (0,-0.7);
  \coordinate (O2) at (0,-3.5);
  \coordinate (O3) at (4,-2.15);
  \draw (O1) circle (0.4);
  \draw (O2) circle (0.4);
  \draw (O3) circle (0.4);
  \draw[black] (2.2,-1.9) ++(-.65*0.5, -7/6*0.5) rectangle ++(0.7,0.7);
   \draw[black] (6.5,-0.45) ++(-.65*0.5, -7/6*0.5) rectangle ++(0.7,0.7);
   \draw[black] (6.5,-3.35) ++(-.65*0.5, -7/6*0.5) rectangle ++(0.7,0.7);
  \draw [dotted] (0,-1.2) -- (0,-3);
  \draw [dotted] (6.5,-1.1) -- (6.5,-3.2);
  \path[black] (-90:0.7) node []{$v_1$};
  \path[black] (-90:3.5) node []{$v_{r-1}$};
  \path[black] (-44:3.05) node []{$c$};
  \path[black] (-28:4.5) node []{$v$};
  \path[black] (-5.75:6.6) node []{$c_1$};
  \path[black] (-29:7.45) node []{$c_{i-1}$};
  \path[black] (-32:3.75) node []{$e$};
  \draw (0.33,-0.92) -- (1.87,-2.15);
  \draw (0.33,-3.27) -- (1.87,-2.15);
  
  \draw (0,-1.71) -- (1.87,-2.15);
  \draw (0,-2.05) -- (1.87,-2.15);
  \draw (0,-2.4) -- (1.87,-2.15);

  \draw (4.4,-2.15) -- (6.5,-2.13);
  \draw (4.4,-2.15) -- (6.5,-2.45);
  \draw (4.4,-2.15) -- (6.5,-1.78);

  \draw (2.6,-2.15) -- (3.6,-2.15);
  \draw (4.283,-1.867) -- (6.17,-0.65);
  \draw (4.283,-2.433) -- (6.17,-3.6);
\end{tikzpicture}
\caption{A tree-like representation of the neighborhood of an edge $e$ between a left node $v$ of degree $i$ and a right node $c$ of degree $r$ in the right-regular bipartite  graph.}\label{fig:denevol1}
\end{figure}
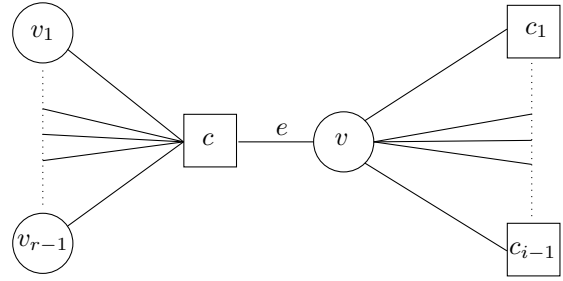 
  
 \begin{lemma}
 Let the probability that a randomly picked item is a defective item and remains unidentified at the end of iteration $j$ of the recovery algorithm be denoted by $p_j$. Also, let the probability that a randomly selected right node is resolved at iteration $j$ of the recovery algorithm be denoted by $q_j$. The following density evolution equations illustrates the relation between $p_j$ and $p_{j+1}$.
 \begin{equation}\label{eq:prob-denevol1}
q_j= \sum_{k=0}^{t-1}{r-1 \choose k}p_j^k(1-p_j)^{r-k-1},
    \end{equation}
    \begin{equation}\label{eq:prob-denevol2}
    p_{j+1}=\gamma \sum_{i=1}^{d} \lambda_i(1-q_j)^{i-1},
 \end{equation}
 where $t$, $r$, $d$, and $\gamma$ are the error correction capability of the BCH code, the degree of right nodes, the maximum degree of left nodes, and the probability that an item is defective, respectively. 
 \end{lemma}
 \begin{proof}
 A tree-like representation of the neighborhood of an edge $e$ between a left node $v$ of degree $i$ and a right node $c$ of degree $r$ is shown in Fig.~\ref{fig:denevol1}. The left node $v$ sends a \textquote{not identified} message to the right node $c$ at iteration $j+1$ through the edge $e$ if none of its other neighboring right nodes $\{c_k\}_{k=1}^{i-1}$ have been resolved at iteration $j$. This event happens with probability $(1-q_j)^{i-1}$. A randomly selected edge is connected to a left node of degree $i$ with probability $\lambda_i$. Thus, a randomly selected left node remains unidentified at the end of iteration $j$ with probability $\sum_{i=1}^{d}\lambda_i(1-q_j)^{i-1}$. Also, we know that each item is defective with probability $\gamma$. Hence, the probability that a randomly picked item is a defective item and remains unidentified at the end of iteration $j$ of the recovery algorithm is given by ${p_{j+1}=\gamma \sum_{i=1}^{d} \lambda_i(1-q_j)^{i-1}}$.

The right node $c$ passes a \textquote{resolved} message to the left node $v$ at iteration $j$ through the edge $e$ if among the other $r-1$ left nodes connected to it only $k\in\{0,1,\cdots,t-1\}$ items are unidentified. This event happens with probability $\sum_{k=0}^{t-1}{r-1 \choose k}p_j^k(1-p_j)^{r-k-1}$. A randomly selected edge is connected to a right node of degree $r$ with probability one. Hence, a randomly selected right node is resolved at iteration $j$ of the decoding algorithm with probability $q_j= \sum_{k=0}^{t-1}{r-1 \choose k}p_j^k(1-p_j)^{r-k-1}$.
  \end{proof}

 The density evolution equations \eqref{eq:prob-denevol1} and \eqref{eq:prob-denevol2} can be combined as 
 \begin{equation}\label{eq:combined_den}
 {p_{j+1}=\gamma \sum_{i=1}^{d}\lambda_i\left(1-\sum_{k=0}^{t-1}{r-1 \choose k}p_j^k(1-p_j)^{r-k-1}\right)^{i-1}}.
 \end{equation}
 Letting $r\rightarrow \infty$ and using the Poisson approximation, the equation \eqref{eq:combined_den} reduces to 
 \begin{equation}\label{eq:norm-dens-evol}
 p_{j+1}= \gamma \sum_{i=1}^{d}\lambda_i\left(1-\sum_{k=0}^{t-1}\frac{(r p_{j})^ke^{-r p_j}}{k!}\right)^{i-1}.
 \end{equation}
 Let $\phi_j \triangleq \frac{p_j}{\gamma}$ and $\psi\triangleq r\gamma$. We can rewrite \eqref{eq:norm-dens-evol} as follows:
 \begin{equation}
\phi_{j+1}=\sum_{i=1}^{d}\lambda_i\left(1-\sum_{k=0}^{t-1}\frac{(\psi \phi_j)^ke^{-\psi \phi_j}}{k!}\right)^{i-1}, 
 \end{equation}
  where $\phi_j$ denotes the probability that a randomly chosen defective item remains unidentified at the end of iteration $j$ of the recovery algorithm.

  The objective is to minimize the total number of tests, $m=M\times s$, where $M$ is the number of right nodes and $s$ is the number of rows in signature matrix. Substituting $\gamma=\frac{K}{N}$ in \eqref{eq:edgeBalance} results in ${M=\frac{\ell}{r\gamma}K}$. Using the fact that $\psi=r\gamma$, we can rewrite the number of right nodes as ${M=\frac{\ell}{\psi}K}$. 
  
  For a given $t$ and $d$, we can minimize the number of right nodes, ${M=\frac{\ell}{\psi}K}$, subject to the constraint ${\phi_{j+1} < \phi_j}$, so as to minimize the total number of the tests. The constraint ${\phi_{j+1} < \phi_j}$ guarantees that $\underset{j\rightarrow \infty}{\lim}\phi_j \rightarrow 0$. In other words, this constraint guarantees that the probability that a randomly selected defective item remains unidentified after running the recovery algorithm for sufficiently large number of iterations, approaches zero. Note that knowing $N$ and $\gamma$ means that $K$ is also known. Thus, the optimization problem reduces to minimizing the fraction $\frac{\ell}{\psi}$. It should be noted that minimizing the fraction $\frac{\ell}{\psi}$ is equivalent to minimizing the fraction $\frac{-\psi}{\ell}$. Using \eqref{eq:avg-lft-dgr}, one can readily see that $\frac{-\psi}{\ell}=-\psi \sum_{i=1}^{d}\frac{\lambda_i}{i}$. We perform a two-step optimization procedure as follows. First, given the parameters $t$ and $d$, we solve the following Linear Programming (LP) problem for any $\psi> 0$.
  
\begin{subequations}
\begin{align}
& \underset{\underset{i \in [d]}{\lambda_i} }{\text{min}}
& & -\psi \sum_{i=1}^{d}\frac{\lambda_i}{i}  \\
& \text{s.t.} & &  \sum_{i=1}^{d}\lambda_i\left(1-\sum_{k=0}^{t-1}\frac{(\psi \phi)^ke^{-\psi \phi}}{k!}\right)^{i-1} < \phi \\
& & &  \sum_{i=1}^{d} \lambda_i =1 \\
& & & \lambda_i \geq  0, \forall i \in [d] 
\end{align}
\end{subequations} 

For any $\psi> 0$, let $f(\psi)\triangleq -\psi \sum_{i=1}^{d}\frac{\lambda_i^{\star}}{i}$, where $\lambda_i^{\star}$'s denote the optimal value of $\lambda_i$'s attained by solving this LP problem. We then minimize $f(\psi)$ over all values of $\psi>0$ as follows.
\begin{equation}\label{eq:min}
\begin{aligned}
& \underset{\psi>0 }{\text{min}}
& & f(\psi)  \\
\end{aligned}
\end{equation}
We can solve this problem numerically and attain the optimal value of $\psi$ which is denoted by $\psi^{\star}$. Let ${c(t,d) \triangleq \frac{-1}{f(\psi^{\star})}}$. Then, the minimum number of right nodes is given by ${M=c(t,d)K}$. Substituting $M=c(t,d)K$ in \eqref{eq:edgeBalance}, one can easily compute ${r=\frac{\ell N}{c(t,d)K}}$. Therefore, the total number tests will become ${m=M\times s=c(t,d)K(t\log (\frac{\ell N}{c(t,d)K}+1)+1)}$.

\subsection{Proof of Theorem~\ref{thm:main2}}
In \cite{DBLP:journals/corr/abs-1901-07635}, there is a typo and the computational complexity presented for the testing algorithm is not correct. Below, we present the correct complexity for the testing algorithm.
The total number of tests is $m=\mathcal{O}(K\log \frac{N}{K})$. For each test, $r$ summations are executed. Thus, the testing algorithm has the computational complexity of $\mathcal{O}(rK\log \frac{N}{K})$. From \eqref{eq:edgeBalance}, one can easily see that $r=\mathcal{O}(\frac{N}{K})$. Then, the computational complexity of the testing algorithm can be stated as $\mathcal{O}(N\log \frac{N}{K})$. 

The total number of right nodes is $M=\mathcal{O}(K)$. The computational complexity of resolving each right node is given by $\mathcal{O}(\log r)$ when $t\leq 4$ (see the proof of \cite[Lemma 4]{DBLP:journals/corr/abs-1901-07635}). Therefore, the computational complexity of the recovery algorithm is $\mathcal{O}(K\log \frac{N}{K})$.  

 \begin{figure}
\includegraphics[width=0.52\textwidth]{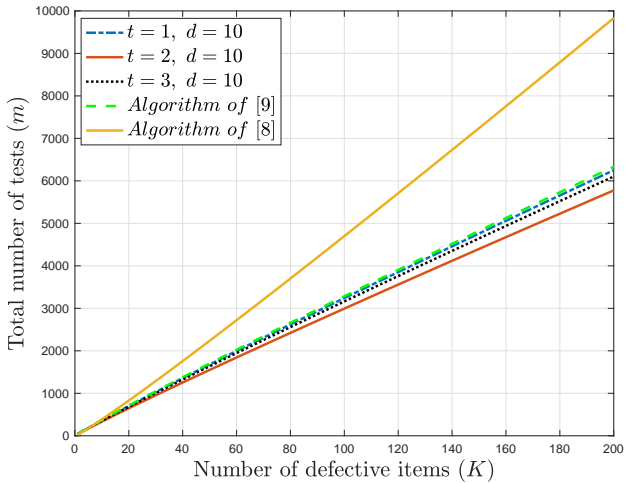}	
\caption{\small The number of required tests ($m$) to identify all defective items (for different values of $K$) among $N=2^{32}$ items obtained via analysis. 
}
\label{fig:best}
\end{figure} 

\section{Comparison Results}
In this section, we evaluate the performance of the proposed scheme via extensive simulations.

We compare the performance of the proposed scheme with the performance of two non-adaptive QGT schemes recently proposed in \cite{gebhard2019quantitative} and \cite{DBLP:journals/corr/abs-1901-07635} based on our theoretical analysis. Fig.~\ref{fig:best} illustrates the total number of tests ($m$) required to identify all defective items. The total number of items is considered to be ${N=2^{32}}$. As it can be seen, the proposed scheme, for $t=2$, requires the minimum number of tests to identify all the defective items. Also, it can be observe that the gap between the proposed scheme and the two other schemes increases as the number of defective items ($K$) grows. 

We also compare the performance of the proposed scheme with the performance of non-adaptive QGT schemes in \cite{gebhard2019quantitative} and \cite{DBLP:journals/corr/abs-1901-07635} using the Monte Carlo simulation. The probability of error, defined as the probability of a defective item to remain unidentified, is depicted in Fig.~\ref{fig:sim} for $K=100$ defective items among a population of $N=2^{16}$ items. For a target error probability, e.g., $10^{-5}$, the required number of tests is minimum for the proposed scheme for $t=3$. 
\begin{figure}
\includegraphics[width=0.52\textwidth]{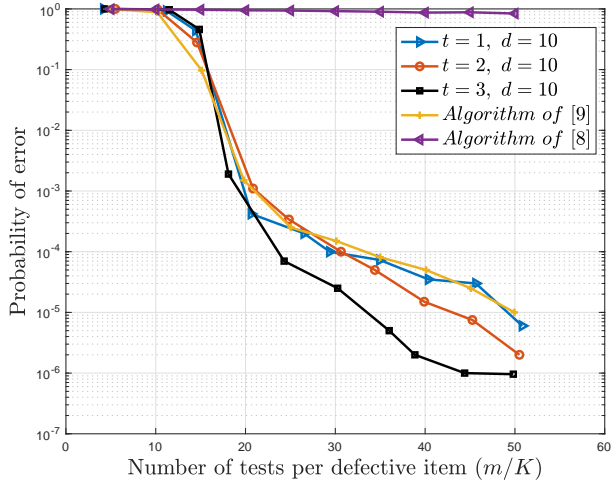}\vspace{-0.2cm}	
\caption{\small The probability of error obtained via Monte Carlo simulations for $N=2^{16}$ items among which $K=100$ items are defective.}\label{fig:sim}\vspace{-0.3cm}
\end{figure}
\bibliographystyle{IEEEtran}
\bibliography{QGTRefs}

\end{document}